\documentclass[12pt,a4paper]{amsart}

\usepackage[hmargin=2.0cm,vmargin=2cm]{geometry}

\usepackage{amsmath,amsfonts,amssymb}

\usepackage{hyperref}
\usepackage{nameref,zref-xr}                    

\newcommand{\Z}{\mathbb Z}
\newcommand{\mbC}{\mathbb C}

\newcommand{\mcM}{\mathcal M}
\newcommand{\oM}{\overline{\mathcal M}}

\newcommand{\half}{\frac{1}{2}}

\newcommand{\<}{\left <}
\renewcommand{\>}{\right >}

\newcommand*{\der}[2]{\frac{\d #1}{\d #2}}
\def\Coef{{\rm Coef}}
\def\d{\partial}
\def\CP1{\mathbb{C}\mathrm{P}^1}

\def\s{\mathfrak{s}}

\newcommand{\cL}{\mathcal L}

\newcommand{\tcL}{\widetilde{\mathcal L}}

\newcommand{\hS}{\widehat S}

\newtheorem{theorem}{Theorem}[section]
\newtheorem{proposition}[theorem]{Proposition}
\newtheorem{lemma}[theorem]{Lemma}
\newtheorem{corollary}[theorem]{Corollary}

\theoremstyle{definition}

\numberwithin{equation}{section}

\title{Open intersection numbers and the wave function of the KdV hierarchy}

\author{A. Buryak}

\address{A.~Buryak:\newline
Department of Mathematics,
ETH Zurich,\newline
Ramistrasse 101 8092, HG G 27.1, Zurich, Switzerland.} 
\email{buryaksh@gmail.com}

\subjclass[2010]{Primary 35Q53; Secondary 14H10}

\keywords{Riemann surfaces with boundary, moduli space, KdV equations}

\begin{document}

\begin{abstract}

Recently R. Pandharipande, J. Solomon and R. Tessler initiated a study of the intersection theory on the moduli space of Riemann surfaces with boundary. They conjectured that the generating series of the intersection numbers is a specific solution of a system of PDEs, that they called the open KdV equations. In this paper we show that the open KdV equations are closely related to the equations for the wave function of the KdV hierarchy. This allows us to give an explicit formula for the specific solution in terms of Witten's generating series of the intersection numbers on the moduli space of stable curves.     

\end{abstract}

\maketitle

\section{Introduction}

Denote by~$\mcM_{g,n}$ the moduli space of smooth complex algebraic curves of genus~$g$ with~$n$ distinct marked points. In~\cite{DM69} P.~Deligne and D.~Mumford defined a natural compactification $\mcM_{g,n}\subset\oM_{g,n}$ via stable curves (with possible nodal singularities). The moduli space $\oM_{g,n}$ is a nonsingular complex orbifold of dimension $3g-3+n$.

A new direction in the study of the moduli space $\oM_{g,n}$ was opened by E. Witten~\cite{Wit91}. The class $\psi_i\in H^2(\oM_{g,n};\mbC)$ is defined as the first Chern class of the line bundle over $\oM_{g,n}$ formed by the cotangent lines at the $i$-th marked point. Intersection numbers $\<\tau_{k_1}\tau_{k_2}\ldots\tau_{k_n}\>^c_g$ are defined as follows:
$$
\<\tau_{k_1}\tau_{k_2}\ldots\tau_{k_n}\>^c_g:=\int_{\oM_{g,n}}\psi_1^{k_1}\psi_2^{k_2}\ldots\psi_n^{k_n}.
$$ 
The superscript $c$ here signals integration over the moduli of closed Riemann surfaces. Let us introduce variables $u,t_0,t_1,t_2,\ldots$ and consider the generating series 
$$
F^c(t_0,t_1,\ldots;u):=\sum_{\substack{g\ge 0,n\ge 1\\2g-2+n>0}}\frac{u^{2g-2}}{n!}\sum_{k_1,\ldots,k_n\ge 0}\<\tau_{k_1}\tau_{k_2}\ldots\tau_{k_n}\>^c_gt_{k_1}t_{k_2}\ldots t_{k_n}.
$$
E.~Witten (\cite{Wit91}) conjectured that the exponent $\exp(F^c)$ is a tau-function of the KdV hierarchy. Witten's conjecture was proved by M.~Kontsevich (\cite{Kon92}). There was a later reformulation of Witten's conjecture due to R.~Dijkgraaf, E.~Verlinde and H.~Verlinde (\cite{DVV91}). They defined certain quadratic differential operators $L_n$, $n\ge -1$, and proved that Witten's conjecture is equivalent to the equations $L_n\exp(F^c)=0$, that are called the Virasoro equations.

In \cite{PST14} the authors initiated a study of the intersection theory on the moduli space of Riemann surfaces with boundary. They introduced intersection numbers on this moduli space and called them the open intersection numbers. The authors completely described them in genus~$0$. In higher genera they conjectured that the generating series of the open intersection numbers satisfies certain partial differential equations that are analogous to the KdV and the Virasoro equations. In~\cite{PST14} these equations were called the open KdV and the open Virasoro equations. 

In \cite{Bur14} we studied the open KdV equations in detail. First, we proved that the system of the open KdV equations together with the initial condition, that corresponds to the simplest open intersection numbers in genus $0$, has a unique solution. We call this solution the open potential $F^o$. Second, in~\cite{Bur14} we showed that the open potential $F^o$ satisfies the open Virasoro equations. Therefore we proved that the open KdV and the open Virasoro equations give equivalent descriptions of the intersection numbers on the moduli space of Riemann surfaces with boundary. This was left in~\cite{PST14} as a conjecture.

In this paper we show that the system of the open KdV equations is closely related to the system of equations for the wave function of the KdV hierarchy. This allows us to give an explicit formula for the open potential $F^o$ in terms of the Witten potential $F^c$.

In \cite{PST14} the authors didn't consider the contangent line bundles at boundary marked points, so the generating series of the open intersection numbers depends only on the variables~$t_0,t_1,\ldots$ and $s$. In \cite{Bur14} we suggested a natural way to introduce new variables $s_1,s_2,\ldots$ in the open potential $F^o$. The resulting power series is called the extended open potential. The new variables~$s_1,s_2,\ldots$ can be viewed as descendants of boundary points. Actually, in this paper we derive an explicit formula for the extended open potential. We also construct Virasoro type equations for the extended open potential. Surprisingly, our new Virasoro equations are even simpler then the open Virasoro equations for the potential $F^o$.   

\subsection{Moduli of Riemann surfaces with boundary}\label{subsection:open moduli}

Here we briefly recall the main definitions from~\cite{PST14}. 

Let $\Delta\in\mbC$ be the open unit disk, and let $\overline\Delta$ be its closure. An extendable embedding of the open disk $\Delta$ in a compact Riemann surface $f\colon\Delta\to C$ is a holomorphic map which can be extended to a holomorphic embedding of an open neighborhood of $\overline\Delta$. Two extendable embeddings are disjoint, if the images of $\overline\Delta$ are disjoint.

A Riemann surface with boundary $(X,\d X)$ is obtained by removing a finite positive number of disjoint extendable open disks from a connected compact Riemann surface. A compact Riemann surface is not viewed here as Riemann surface with boundary.

To a Riemann surface with boundary $(X,\d X)$, we can canonically construct the double via the Schwartz reflection through the boundary. The double $D(X,\d X)$ of $(X,\d X)$ is a compact Riemann surface. The doubled genus of $(X,\d X)$ is defined to be the usual genus of $D(X,\d X)$.

On a Riemann surface with boundary $(X,\d X)$, we consider two types of marked points. The markings of interior type 
are points of $X\backslash\d X$. The markings of boundary type are points of~$\d X$. Let $\mcM_{g,k,l}$ denote the moduli space of Riemann surfaces with boundary of doubled genus~$g$ with~$k$ distinct boundary markings and $l$ distinct interior markings. The moduli space~$\mcM_{g,k,l}$ is defined to be empty  unless the stability condition 
$$
2g-2+k+2l>0
$$
is satisfied. The moduli space $\mcM_{g,k,l}$ is a real orbifold of real dimension $3g-3+k+2l$. 

The cotangent line classes $\psi_i\in H^2(\mcM_{g,k,l};\mbC)$ are defined (as before) as the first Chern classes of the cotangent line bundles associated to the interior markings. In \cite{PST14}, cotangent lines at the boundary points are not considered. Open intersection numbers are defined by
\begin{gather}\label{eq: open intersections}
\left<\tau_{a_1}\tau_{a_2}\ldots\tau_{a_l}\sigma^k\right>^o_g:=\int_{\oM_{g,k,l}}\psi_1^{a_1}\psi_2^{a_2}\ldots\psi_l^{a_l}.
\end{gather}
To rigorously define the right-hand side of~\eqref{eq: open intersections}, at least three significant steps must be taken:
\begin{itemize}

\item A natural compactification $\mcM_{g,k,l}\subset\oM_{g,k,l}$ must be constructed. Candidates for $\oM_{g,k,l}$ are themselves real orbifolds with boundary $\d\oM_{g,k,l}$;

\item For integration over $\oM_{g,k,l}$ to be well-defined, boundary conditions of the integrand along $\d\oM_{g,k,l}$ must be specified;

\item Orientation issues should be resolved, since the moduli space $\mcM_{g,k,l}$ is in general non-orientable.  

\end{itemize}
All three steps are completed in genus~$0$ in~\cite{PST14}. The higher genus constructions will appear in upcoming work of J.~P.~Solomon and R.~J.~Tessler~\cite{ST}.

\subsection{Open KdV equations}

Here we recall the open KdV equations and formulate the main conjecture from~\cite{PST14}. 

Let~$F$ be a formal power series in the variables $t_0,t_1,\ldots$ and~$s$ with the coefficients from~$\mbC[u,u^{-1}]$. The open KdV equations are the following partial differential equations:
\begin{gather}\label{eq:open kdv}
\frac{2n+1}{2}\frac{\d F}{\d t_n}=u\frac{\d F}{\d s}\frac{\d F}{\d t_{n-1}}+u\frac{\d^2 F}{\d s\d t_{n-1}}+\frac{u^2}{2}\frac{\d F}{\d t_0}\frac{\d^2 F^c}{\d t_0\d t_{n-1}}-\frac{u^2}{4}\frac{\d^3 F^c}{\d t_0^2\d t_{n-1}},\quad n\ge 1.
\end{gather}
In \cite{PST14} the authors conjectured, that the generating series 
$$
F^{o,geom}(t_0,t_1,\ldots,s;u):=\sum_{\substack{g,k,l\ge 0\\2g-2+k+2l>0}}\sum_{a_1,\ldots,a_l\ge 0}\frac{u^{g-1}}{k!l!}\left<\tau_{a_1}\tau_{a_2}\ldots\tau_{a_l}\sigma^k\right>^o_g s^k \prod_{i=1}^lt_{a_i}
$$
satisfies the open KdV equations. The following initial condition follows easily from the definitions:
\begin{gather}\label{eq:initial conditions}
\left.F^{o,geom}\right|_{t_{\ge 1}=0}=u^{-1}\left(\frac{s^3}{6}+t_0 s\right).
\end{gather}
It is clear that the open KdV equations together with this initial condition and the closed potential~$F^c$ completely determine the power series~$F^{o,geom}$. Denote by~$F^o$ a unique solution of the system~\eqref{eq:open kdv}, that satisfies the initial condition~\eqref{eq:initial conditions}. The existence of such a solution is non-trivial and was proved in~\cite{Bur14}. The power series~$F^o$ will be called the open potential. So the main conjecture from~\cite{PST14} can be formulated as the equation~$F^{o,geom}=F^o$.

In \cite{Bur14} we proved, that the power series $F^o$ satisfies also the following equation:
$$
\der{F^o}{s}=u\left(\frac{1}{2}\left(\der{F^o}{t_0}\right)^2+\frac{1}{2}\frac{\d^2 F^o}{\d t_0^2}+\frac{\d^2 F^c}{\d t_0^2}\right).
$$
This equation allows us to eliminate the derivatives by $s$ on the right-hand side of~\eqref{eq:open kdv}. Therefore the potential $F^o$ has the following equivalent description. It is a unique solution of the system
\begin{align}
F_s&=u\left(\frac{1}{2}F_0^2+\frac{1}{2}F_{0,0}+F^c_{0,0}\right),\label{eq:s-equation}\\
\frac{2n+1}{2u^2}F_n&=\left(\frac{1}{2}\frac{\d^2}{\d t_0^2}+F_0\frac{\d}{\d t_0}+\frac{1}{2}F_0^2+\frac{1}{2}F_{0,0}+F^c_{0,0}\right)F_{n-1}+\frac{1}{2}F_0 F^c_{0,n-1}+\frac{3}{4}F^c_{0,0,n-1},\,n\ge 1,\label{eq:first half}
\end{align}
with the initial condition 
\begin{gather}\label{eq:initial condition2}
\left.F\right|_{\substack{t_{\ge 1}=0\\s=0}}=0.
\end{gather}
Here we use the subscript $n$ for the partial derivative by $t_n$. 

Consider variables $s_0,s_1,s_2,\ldots$ and let $s_0:=s$. Let $F$ be a formal power series in the variables~$t_0,t_1,\ldots$,~$s_0,s_1,\ldots$ with the coefficients from $\mbC[u,u^{-1}]$. In \cite{Bur14} we introduced the following equations:
\begin{gather}\label{eq:second half}
\frac{n+1}{u^2}F_{s_n}=\left(\frac{1}{2}\frac{\d^2}{\d t_0^2}+F_0\frac{\d}{\d t_0}+\frac{1}{2}F_0^2+\frac{1}{2}F_{0,0}+F^c_{0,0}\right)F_{s_{n-1}},\quad n\ge 1.
\end{gather}
We proved that the system \eqref{eq:s-equation}-\eqref{eq:first half} together with \eqref{eq:second half} has a unique solution for an arbitrary initial condition of the form $\left.F\right|_{\substack{t_{\ge 1}=0\\s_*=0}}=p(t_0)$, where $p$ is a polynomial in $t_0,u,u^{-1}$. Let us denote by~$F^{o,ext}$ a unique solution of this big system, that is specified by the initial condition
\begin{gather}\label{eq:extended initial condition}
\left.F\right|_{\substack{t_{\ge 1}=0\\s_*=0}}=0.
\end{gather}
We call this solution the extended open potential. We obviously have 
$$
\left.F^{o,ext}\right|_{s_{\ge 1}=0}=F^o.
$$

The main observation of this paper is that the system of equations \eqref{eq:s-equation},\eqref{eq:first half} and~\eqref{eq:second half} is in fact not new. If we rewrite it for the exponent $\exp(F)$, we get the system that coincides with the system of equations for the wave function of the KdV hierarchy. This allows us to derive an explicit formula for~$F^{o,ext}$ that we present in the next section.


\subsection{Explicit formula for $F^{o,ext}$}\label{subsection:explicit formula}

In this section we formulate the main result of the paper -- an explicit formula for the potential $F^{o,ext}$. 

Denote by $G_z$ the shift operator that acts on a power series $f(t_0,t_1,\ldots)\in\mbC[u,u^{-1}][[t_0,t_1,\ldots]]$ as follows:
\begin{gather*}
G_z(f)(t_0,t_1,\ldots):=f\left(t_0-\frac{k_0}{z},t_1-\frac{k_1}{u^2 z^3},t_2-\frac{k_2}{u^4z^5},\ldots\right),
\end{gather*}
where $k_n:=(2n-1)!!$ and, by definition, we put $(-1)!!:=1$. Let 
$$
\xi(t_0,t_1,\ldots,s_0,s_1,\ldots;z):=\sum_{i\ge 0}\frac{t_i u^{2i} z^{2i+1}}{(2i+1)!!}+\sum_{i\ge 0}\frac{s_i u^{2i+1} z^{2i+2}}{2^{i+1}(i+1)!}.
$$
Let us also introduce a Laurent series $D(z)\in\mbC[u,u^{-1}][[z^{-1}]]$. Define the numbers $a_0,a_1,\ldots$ and $d_0,d_1,d_2,\ldots$ by
\begin{align}
&a_n:=(-1)^n\frac{(6n)!}{288^n(2n)!(3n)!},\notag\\
&d_n:=\sum_{i=0}^n 3^i|a_{n-i}|\prod_{k=1}^i\left(n+\half-k\right)\label{eq:definition of D}.
\end{align}
The Laurent series $D(z)$ is defined by $D(z):=1+\sum_{i\ge 1}\frac{d_i}{u^{2i} z^{3i}}$. Now we are ready to formulate our main theorem.

\begin{theorem}\label{theorem:main}
We have
\begin{gather}\label{eq:main formula}
\exp(F^{o,ext})=\Coef_{z^0}\left[D(z)\frac{G_z(\exp\left(F^c)\right)}{\exp(F^c)}\exp(\xi)\right].
\end{gather}
\end{theorem}
\noindent We see that $D(z)\frac{G_z(\exp\left(F^c)\right)}{\exp(F^c)}$ is a power series in $z^{-1}$. On the other hand,~$\exp(\xi)$ is a power series in $z$. In general one can have problems with the multiplication of two such series. In our case the issue can be solved as follows. Let us introduce a gradation in the ring $\mbC[u,u^{-1}][[t_0,t_1,\ldots,s_0,s_1,\ldots]]$ assigning to $t_i$ degree $2i+1$ and to $s_i$ degree $2i+2$. Since the degree of the coefficient of~$z^i$ in~$\exp(\xi)$ grows as~$i$ grows, the product in the square brackets on the right-hand side of~\eqref{eq:main formula} is well-defined.


\subsection{Extended open Virasoro equations}

Here we recall the open Virasoro equations from~\cite{PST14} and construct Virasoro type equations for the extended potential~$F^{o,ext}$.

The Virasoro operators $L_n,n\ge -1$, are defined as follows:
\begin{multline*}
L_n:=\sum_{i\ge 0}\frac{(2i+2n+1)!!}{2^{n+1}(2i-1)!!}(t_i-\delta_{i,1})\frac{\d}{\d t_{i+n}}+\frac{u^2}{2}\sum_{i=0}^{n-1}\frac{(2i+1)!!(2n-2i-1)!!}{2^{n+1}}\frac{\d^2}{\d t_i\d t_{n-1-i}}\\
+\delta_{n,-1}\frac{t_0^2}{2 u^2}+\delta_{n,0}\frac{1}{16}.
\end{multline*}
The Witten potential $F^c$ satisfies the so-called Virasoro equations (see \cite{DVV91}): 
\begin{gather}\label{eq:closed virasoro}
L_n\exp(F^c)=0,\quad n\ge -1.
\end{gather}
The equation $L_{-1}\tau^c=0$ is called the string equation.

In \cite{PST14} the authors introduced the following modified operators:
\begin{gather*}
\cL_n:=L_n+u^n s\frac{\d^{n+1}}{\d s^{n+1}}+\frac{3n+3}{4}u^n\frac{\d^n}{\d s^n},\quad n\ge -1,
\end{gather*}
and conjectured that
\begin{gather}\label{eq:open Virasoro equations}
\cL_n\exp(F^o+F^c)=0,\quad n\ge -1.
\end{gather}
These equations are called the open Virasoro equations. We proved this conjecture in \cite{Bur14}.

Introduce operators $\cL_n^{ext}$, $n\ge -1$, by
$$
\cL_n^{ext}:=L_n+\sum_{i\ge 0}\frac{(i+n+1)!}{i!}s_i\frac{\d}{\d s_{n+i}}+\frac{3(n+1)!}{4}u\frac{\d}{\d s_{n-1}}+\delta_{n,-1}u^{-1}s+\delta_{n,0}\frac{3}{4}.
$$
Here we, by definition, put $\frac{\d}{\d s_{-2}}:=\frac{\d}{\d s_{-1}}:=0$. Our second result is the following theorem.
\begin{theorem}\label{theorem:virasoro}
For any $n\ge -1$ we have 
\begin{gather}\label{eq:new virasoro}
\cL_n^{ext}\exp(F^{o,ext}+F^c)=0.
\end{gather}
\end{theorem}
\noindent We call these equations the extended open Virasoro equations. At first glance, our new Virasoro operators $\cL_n^{ext}$ look quite different from the operators $\cL_n$. Actually, we will show that the open Virasoro equations~\eqref{eq:open Virasoro equations} can be easily derived from the extended open Virasoro equations~\eqref{eq:new virasoro}.


\subsection{Organization of the paper}

In Section~\ref{section:kdv hierarchy} we recall all necessary facts about the KdV hierarchy. In Section~\ref{section:open kdv and wave function} we show that the exponent $\exp\left(F^{o,ext}\right)$ satisfies the same equations as the wave function of the KdV hierarchy. Section~\ref{section:main theorem} is devoted to the proof of Theorem~\ref{theorem:main}. In Section~\ref{section:virasoro equations} we prove Theorem~\ref{theorem:virasoro}. Appendix~\ref{section:technical lemmas} contains some technical statements that we use in Sections~\ref{section:open kdv and wave function} and \ref{section:main theorem}.


\subsection{Acknowledgments}

We thank G. Carlet, R. Pandharipande, S. Shadrin, J. Solomon and R. Tessler for discussions related to the work presented here.

We would like to thank the anonymous referee for valuable remarks and suggestions that allowed us to improve the exposition of this paper.

This work was supported by grant ERC-2012-AdG-320368-MCSK in the group of R. Pandharipande at ETH Zurich, by the Russian Federation Government grant no. 2010-220-01-077 (ag. no. 11.634.31.0005), the grants RFFI 13-01-00755 and NSh-4850.2012.1.

Part of the work was completed during the visit of the author to the Einstein Institute of Mathematics of the Hebrew University of Jerusalem in 2014.

\subsection{Note added in proof}

We would like to mention several papers~\cite{Ale14a,Ale14b,BY14,BH15,BT15} that are related to the present work and that appeared while this paper was under consideration in the journal. In \cite{Ale14b} (see also~\cite{Ale14a}) A. Alexandrov used Theorem~\ref{theorem:main} and derived a matrix model for the exponent~$e^{F^{o,ext}+F^c}$. He also proved that after a certain simple change of variables the exponent~$e^{F^{o,ext}+F^c}$ becomes a tau-function of the KP hierarchy. In \cite{BY14} M.~Bertola and D.~Yang suggested a natural candidate for the generating series of open $r$-spin intersection numbers and found a generalization of our Theorem~\ref{theorem:main}. In~\cite{BT15}, together with R.~Tessler we proved the main conjecture from~\cite{PST14} that the generating series of the open intersection numbers satisfies the open KdV equations.


\section{Lax form of the KdV hierarchy}\label{section:kdv hierarchy}

Here we recall the Lax formalism in the theory of the KdV hierarchy. The book~\cite{Dic03} is a very good reference to this subject. Our notations are slightly different, because we rescale the flows of the KdV hierarchy and also we insert the formal variable $u$ in the coefficients of the equations of the hierarchy. Similar notations are used in~\cite{CvdLS14}. 

\subsection{Pseudo-differential operators}

A pseudo-differential operator $A$ is a Laurent series 
$$
A=\sum_{n=-\infty}^m a_n(t)\d_x^n,
$$
where $m$ is an arbitrary integer and $a_n$ are formal power series in $t_0,t_1,\ldots$ with the coefficients from $\mbC[u,u^{-1}]$. We will always identify $t_0$ with $x$. Let
\begin{gather*}
A_+:=\sum_{n=0}^m a_n\d_x^n,\quad\text{and}\quad A_-:=A-A_+.
\end{gather*}
The product of pseudo-differential operators is defined by the following commutation rule:
\begin{gather*}
\d_x^k\circ f:=\sum_{l=0}^\infty\frac{k(k-1)\ldots(k-l+1)}{l!}\frac{\d^l f}{\d x^l}\d_x^{k-l},
\end{gather*} 
where $k\in\Z$ and $f\in\mbC[u,u^{-1}][[t_0,t_1,\ldots]]$. The conjugate $A^*$ of a pseudo-differential operator $A=\sum_{n=-\infty}^m a_n\d_x^n$ is defined by
$$
A^*:=\sum_{n=-\infty}^m (-1)^n\d_x^n\circ a_n.
$$

It is not hard to check that, for any pseudo-differential operator $P$ of the form $P=1+\sum_{n\ge 1}p_n(t)\d_x^{-n}$, there exists a unique inverse operator $P^{-1}$ of the same form $P^{-1}=1+\sum_{n\ge 1}\widetilde p_n(t)\d_x^{-n}$. 

Let $m\ge 2$. Consider a pseudo-differential operator $A$ of the form 
$$
A=\d_x^m+\sum_{n=1}^\infty a_n\d_x^{m-n}.
$$
It is not hard to see that there exists a unique pseudo-differential operator $A^{\frac{1}{m}}$ of the form
$$
A^{\frac{1}{m}}=\d_x+\sum_{n=1}^\infty \widetilde{a}_n\d_x^{m-n},
$$
such that $\left(A^{\frac{1}{m}}\right)^m=A$. This operator is called the $m$-th root of $A$.


\subsection{Lax form of the KdV hierarchy}

Let 
$$
L:=\d_x^2+2w,\quad w\in\mbC[u,u^{-1}][[t_0,t_1,\ldots]].
$$
It is easy to check that, for any $n\ge 1$, the commutator $\left[\left(L^{n+\frac{1}{2}}\right)_+,L\right]$ doesn't have terms with positive powers of $\d_x$. The KdV hierarchy is the following system of partial differential equations for the power series $w$:
\begin{gather}\label{eq:Lax equations}
\frac{\d L}{\d t_n}=\frac{u^{2n}}{(2n+1)!!}\left[\left(L^{n+\frac{1}{2}}\right)_+,L\right],\quad n=1,2,\ldots
\end{gather}
This form of the KdV equations is called the Lax form.


\subsection{Dressing operator}

Let $L$ be a solution of the Lax equations~\eqref{eq:Lax equations}. Then there exists a pseudo-differential operator $P$ of the form
\begin{gather}\label{dressing operator}
P=1+\sum_{n\ge 1}p_n(t)\d_x^{-n},
\end{gather}
such that
\begin{align}
&L=P\circ\d_x^2\circ P^{-1},\quad\text{and}\label{L from dressing operator}\\
&\frac{\d P}{\d t_n}=-\frac{u^{2n}}{(2n+1)!!}\left(L^{n+\frac{1}{2}}\right)_-\circ P,\quad n=1,2\ldots.\label{eq:Sato-Wilson equations}
\end{align}
The operator~$P$ is called the dressing operator and the Laurent series 
$$
\widehat P(t;z):=1+\sum_{n\ge 1}p_n(t)z^{-n},
$$
is called the symbol of the dressing operator $P$. Equations~\eqref{eq:Sato-Wilson equations} are called the Sato-Wilson equations. 

We also have the following converse statement. Let $P$ be a pseudo-differential operator of the form~\eqref{dressing operator}. Suppose that the operator $P$ satisfies the equations
$$
\frac{\d P}{\d t_n}=-\frac{u^{2n}}{(2n+1)!!}\left(P\circ\d_x^{2n+1}\circ P^{-1}\right)_-\circ P,\quad n=1,2\ldots,
$$
and has the property $\left(P\circ\d_x^2\circ P^{-1}\right)_-=0$. Define an operator~$L$ by~\eqref{L from dressing operator}. Then the operator~$L$ has the form $L=\d_x^2+2w$ for some power series $w\in\mbC[u,u^{-1}][[t_0,t_1,\ldots]]$ and it satisfies the Lax equations~\eqref{eq:Lax equations}. 


\subsection{Wave function}

Let $P=1+\sum_{n\ge 1}p_n(t)\d_x^{-n}$ be the dressing operator for some solution~$L$ of the Lax equations~\eqref{eq:Lax equations}. It is convenient to introduce variables $t_n$ with $n\in\frac{1}{2}+\Z_{\ge 0}$ and put $t_{\frac{1}{2}+n}:=s_n$. Then the power series $\xi(t_*,s_*;z)$, that we defined in Section~\ref{subsection:explicit formula}, can be written as follows:
$$
\xi(t;z)=\sum_{i\in\frac{1}{2}\Z_{\ge 0}}\frac{t_i u^{2i} z^{2i+1}}{(2i+1)!!}.
$$
Here $(2n)!!:=2^n n!$. The product 
$$
\psi(t;z):=\widehat P(t;z)\exp(\xi(t;z))
$$
is called the wave function of the KdV hierarchy. It satisfies the following equations:
\begin{align}
\frac{\d\psi}{\d t_n}&=\frac{u^{2n}}{(2n+1)!!}\left(L^{n+\frac{1}{2}}\right)_+\psi,\quad n\in\half\Z_{\ge 1};\label{eq:equations for the wave function}\\
L\psi&=z^2\psi.\label{eq:second equation for the wave function}
\end{align}


\subsection{Tau-function}

Let~$L$ be a solution of the Lax equations~\eqref{eq:Lax equations}. Then there exists a power series $\tau\in\mbC[u,u^{-1}][[t_0,t_1,\ldots]]$ such that $\left.\tau\right|_{t_*=0}=1$ and the power series 
$$
\frac{G_z(\tau)}{\tau}\in\mbC[u,u^{-1}][[z^{-1},t_0,t_1,\ldots]]
$$
is the symbol of the dressing operator for the operator $L$. The power series $\tau$ is called the tau-function of the KdV hierarchy. The tau-function has the following property:
$$
\frac{\d^2\ln\tau}{\d t_0^2}=w.
$$


\subsection{Witten's conjecture}

Let $L=\d_x^2+2 F^c_{0,0}$. Witten's conjecture says that the operator~$L$ is a solution of the Lax equations~\eqref{eq:Lax equations} and the exponent~$\exp(F^c)$ is the corresponding tau-function.


\section{Open KdV equations and the equations for the wave function of the KdV hierarchy}\label{section:open kdv and wave function}

In this section we show that the exponent $\exp(F^{o,ext})$ satisfies equations~\eqref{eq:equations for the wave function} for the wave function of the KdV hierarchy.

Let $w:=F^c_{0,0}$ and $L:=\d_x^2+2w$. By Witten's conjecture, the operator $L$ is a solution of the Lax equations~\eqref{eq:Lax equations} and the exponent~$\exp(F^c)$ is the corresponding tau-function. Recall that we have introduced the variables $t_n$ with $n\in\half+\Z_{\ge 0}$ and put $t_{k+\half}:=s_k$.

\begin{theorem}\label{theorem:open kdv and the wave function}
For any $n\in\frac{1}{2}\Z_{\ge 0}$ we have 
\begin{gather}\label{eq:open kdv and the wave function}
\frac{\d}{\d t_n}\exp(F^{o,ext})=\frac{u^{2n}}{(2n+1)!!}\left(L^{n+\frac{1}{2}}\right)_+\exp(F^{o,ext}).
\end{gather}
\end{theorem}
\noindent Recall that we identify $x$ with $t_0$. 

\begin{proof}

The proof is by induction on $n$. Suppose $n$ is a half-integer, $n=k+\half$. If $k=0$, then, by~\eqref{eq:s-equation},
$$
\der{}{t_\half}F^{o,ext}=\der{}{s}F^{o,ext}=u\left(\half F^{o,ext}_{0,0}+\half (F^{o,ext}_0)^2+w\right).
$$
Therefore,
\begin{gather}\label{eq:s-equation for exp}
\der{}{t_{\half}}\exp(F^{o,ext})=\frac{u}{2}\left(\d_x^2+2w\right)\exp(F^{o,ext}).
\end{gather}
This is exactly equation~\eqref{eq:open kdv and the wave function} for $n=\half$.

Suppose~$k$ is a positive integer. We compute that
\begin{align*}
\frac{k+1}{u^2}\der{}{t_{k+\half}}&\exp(F^{o,ext})=\frac{k+1}{u^2}\der{}{s_k}\exp(F^{o,ext})=\frac{k+1}{u^2}F^{o,ext}_{s_k}\exp(F^{o,ext})\stackrel{\text{by \eqref{eq:second half}}}{=}\\
=&\left[\left(\frac{1}{2}\frac{\d^2}{\d t_0^2}+F^{o,ext}_0\frac{\d}{\d t_0}+\frac{1}{2}\left(F^{o,ext}_0\right)^2+\frac{1}{2}F^{o,ext}_{0,0}+w\right)F^{o,ext}_{s_{k-1}}\right]\exp(F^{o,ext})\stackrel{\text{by \eqref{eq:s-equation}}}{=}\\
=&u^{-1}\left(F^{o,ext}_{s,s_{k-1}}+F^{o,ext}_sF^{o,ext}_{s_{k-1}}\right)\exp(F^{o,ext})=u^{-1}\frac{\d^2}{\d s\d s_{k-1}}\exp(F^{o,ext})\stackrel{\substack{\text{by the induction}\\\text{assumption}}}{=}\\
=&u^{-1}\frac{\d}{\d s}\left[\frac{u^{2k-1}}{(2k)!!}L^k\exp(F^{o,ext})\right]\stackrel{\text{by~\eqref{eq:s-equation for exp}}}{=}\frac{u^{2k-1}}{2(2k)!!}L^{k+1}\exp(F^{o,ext}).
\end{align*}
Thus, equation~\eqref{eq:open kdv and the wave function} is proved for half-integers.

Suppose $n$ is an integer. For $n=0$ equation~\eqref{eq:open kdv and the wave function} is obvious. Suppose that $n\ge 1$. We have
\begin{align*}
\frac{2n+1}{2}\der{}{t_n}&\exp(F^{o,ext})=\frac{2n+1}{2}F^{o,ext}_n\exp(F^{o,ext})\stackrel{\text{by \eqref{eq:open kdv}}}{=}\\
=&\left[uF^{o,ext}_{s,n-1}+u F^{o,ext}_s F^{o,ext}_{n-1}+\frac{u^2}{2}F^c_{0,n-1}F^{o,ext}_0-\frac{u^2}{4}F^c_{0,0,n-1}\right]\exp(F^{o,ext})=\\
=&\left(u\frac{\d^2}{\d s\d t_{n-1}}+\frac{u^2}{2}F^c_{0,n-1}\frac{\d}{\d t_0}-\frac{u^2}{4}F^c_{0,0,n-1}\right)\exp(F^{o,ext})\stackrel{\substack{\text{by the induction}\\\text{assumption and \eqref{eq:s-equation for exp}}}}{=}\\
=&\left(\frac{u^{2n}}{2(2n-1)!!}(L^{n-\half})_+\circ L+\frac{u^2}{2}F^c_{0,n-1}\d_x-\frac{u^2}{4}F^c_{0,0,n-1}\right)\exp(F^{o,ext})\stackrel{\text{by Lemma~\ref{lemma:recursion lemma}}}{=}\\
=&\frac{u^{2n}}{2(2n-1)!!}(L^{n+\half})_+\exp(F^{o,ext}).
\end{align*}
The theorem is proved.
\end{proof}


\section{Proof of Theorem~\ref{theorem:main}}\label{section:main theorem}

In this section we prove Theorem~\ref{theorem:main} and show that $\cL_{-1}^{ext}\exp(F^{o,ext}+F^c)=0$.

Let $\tau^c:=\exp(F^c)$. By Witten's conjecture, $\tau^c$ is a tau-function of the KdV hierarchy. Therefore the series $\psi$, defined by 
\begin{gather*}
\psi:=\frac{G_z(\tau^c)}{\tau^c}\exp\left(\xi\right),
\end{gather*}
is the wave function of the KdV hierarchy. Let $w:=F^c_{0,0}$, $L:=\d_x^2+2w$ and 
$$
H(t;z):=G_z(\tau^c)\exp(\xi)=\tau^c\psi.
$$
The proof of Theorem~\ref{theorem:main} is based on the following proposition.
\begin{proposition}\label{proposition:main proposition}
We have $\cL^{ext}_{-1}\Coef_{z^0}(D\cdot H)=0$.
\end{proposition}
\begin{proof}
It is easy to compute that
\begin{align*}
\cL_{-1}^{ext}H=&G_z(L_{-1}\tau^c)\exp(\xi)+\\
&+\left(\sum_{i\ge 0}\frac{(2i+1)!!}{u^{2i+2}z^{2i+3}}G_z(F^c_i)+\sum_{i\ge 0}t_i\frac{u^{2i-2}z^{2i-1}}{(2i-1)!!}+\sum_{i\ge 0}s_i\frac{u^{2i-1}z^{2i}}{2^i i!}-\frac{1}{2u^2z^2}-z\right)H,\\
\frac{1}{u^2 z}\der{}{z}H=&\left(\sum_{i\ge 0}\frac{(2i+1)!!}{u^{2i+2}z^{2i+3}}G_z(F^c_i)+\sum_{i\ge 0}t_i\frac{u^{2i-2}z^{2i-1}}{(2i-1)!!}+\sum_{i\ge 0}s_i\frac{u^{2i-1}z^{2i}}{2^i i!}\right)H.
\end{align*}
Since~$L_{-1}\tau^c=0$, we get
\begin{gather}\label{eq:tmp4}
\frac{\cL_{-1}^{ext}H}{\tau^c}=(\tau^c)^{-1}\left(\frac{1}{u^2 z}\der{}{z}-\frac{1}{2u^2z^2}-z\right)H=\left(\frac{1}{u^2 z}\der{}{z}-\frac{1}{2u^2z^2}-z\right)\psi.
\end{gather}
Denote the operator $\frac{1}{u^2 z}\der{}{z}-\frac{1}{2u^2z^2}-z$ by $S_z$. We see that we have to prove that
\begin{gather}\label{eq:zero equality}
\Coef_{z^0}(D\cdot S_z\psi)=0.
\end{gather}

Let $\mathcal F\in\mbC[u,u^{-1}][[t_0,t_\half,t_1,t_{\frac{3}{2}},\ldots]]$. Consider the following system:
\begin{gather}\label{eq:equations for F}
\frac{\d\mathcal F}{\d t_n}=\frac{u^{2n}}{(2n+1)!!}\left(L^{n+\half}\right)_+\mathcal F,\quad n\in\half\Z_{\ge 1}.
\end{gather}
From \eqref{eq:equations for the wave function} it follows that the series $\Coef_{z^0}(D\cdot S_z\psi)$ is a solution of this system. Therefore, in order to prove~\eqref{eq:zero equality}, it is sufficient to prove that
$$
\left.\Coef_{z^0}(D\cdot S_z\psi)\right|_{t_{>0}=0}=0.
$$
Let $y:=z^{-1}$ and $f(x,y):=\left.\psi\right|_{\substack{t_{>0}=0\\z=y^{-1}}}$. Since $\left.F^c\right|_{t_{>0}=0}=\frac{x^3}{6u^2}$, we have 
$$
f(x,y)=\frac{\exp\left(F^c\left(x-y,-\frac{y^3}{u^2},-\frac{3!!y^5}{u^4},\ldots\right)\right)}{\exp\left(\frac{x^3}{6u^2}\right)}\exp\left(\frac{x}{y}\right).
$$
Denote by $S_y$ the operator $-S_z$, rewritten in the $y$-coordinate. We have
$$
S_y=\frac{y^3}{u^2}\d_y+\frac{y^2}{2u^2}+\frac{1}{y}.
$$
We have to prove that 
\begin{gather}\label{eq:zero equality2}
\Coef_{y^0}(D\cdot S_y f)=0.
\end{gather}
In the next three lemmas we describe several properties of the series $f$.
\begin{lemma}\label{lemma:homogeneity}
The series $f$ satisfies the homogeneity condition $\left(x\der{}{x}+y\der{}{y}+\frac{3}{2}u\der{}{u}\right)f=0$.
\end{lemma}
\begin{proof}
Let $O:=x\der{}{x}+y\der{}{y}+\frac{3}{2}u\der{}{u}$. A correlator $\<\tau_{d_1}\ldots\tau_{d_n}\>^c_g$ can be non-zero, only if $d_1+\ldots+d_n=3g-3+n$. Therefore we have 
$$
\left(\sum_{i\ge 0}(1-i)t_i\der{}{t_i}+\frac{3}{2}u\der{}{u}\right)F^c=0.
$$
This equation easily implies that $O F^c(x-y,-\frac{y^3}{u^2},-\frac{3!!y^5}{u^4},\ldots)=0$. Since $O(\frac{x^3}{6u^2})=O(\frac{x}{y})=0$, we get $O f=0$.
\end{proof}

\begin{lemma}\label{lemma:dx-derivative}
We have $\frac{\d f}{\d x}=S_y f$.
\end{lemma}
\begin{proof}
By~\eqref{eq:tmp4}, $(\tau^c)^{-1}\cL_{-1}^{ext}H=S_z\psi$. It is easy to see that 
\begin{multline*}
\left.\left((\tau^c)^{-1}\cL_{-1}^{ext}H\right)\right|_{t_{>0}=0}=e^{-\frac{t_0^3}{6u^2}}\left(-\frac{\d}{\d t_0}+\frac{t_0^2}{2u^2}\right)\left(\left.H\right|_{t_{>0}=0}\right)=\\
=-\frac{\d}{\d t_0}\left(e^{-\frac{t_0^3}{6u^2}}\left.H\right|_{t_{>0}=0}\right)=-\left.\left(\frac{\d\psi}{\d t_0}\right)\right|_{t_{>0}=0}.
\end{multline*}
Therefore, $\frac{\d f}{\d x}=S_y f$.
\end{proof}

Now we want to compute $\left.f\right|_{x=0}$. Recall that the numbers $a_n$, $n\ge 0$, are defined by
$$
a_n:=(-1)^n\frac{(6n)!}{288^n(2n)!(3n)!}.
$$
Introduce the power series $A(y)$ by
$$
A(y):=1+\sum_{i\ge 1}a_i\frac{y^{3i}}{u^{2i}}.
$$

\begin{lemma}\label{lemma:specialization}
We have $\left.f\right|_{x=0}=A(y)$.
\end{lemma}
\begin{proof}
It is possible to get the proof from the results of \cite{Kon92} or \cite{KS91}, but we decide to present a simple direct argument here. First of all, from Lemma~\ref{lemma:homogeneity} it follows that the series~$\left.f\right|_{x=0}$ has the form 
$$
\left.f\right|_{x=0}=1+\sum_{i\ge 1}f_i\frac{y^{3i}}{u^{2i}},
$$
where $f_i$ are some complex coefficients. From Lemma~\ref{lemma:dx-derivative} it follows that $\frac{\d^2}{\d x^2}f=S_y^2 f$. By~\eqref{eq:second equation for the wave function}, 
$$
(\d_x^2+2x)f=\frac{1}{y^2}f.
$$
So we get 
$$
S_y^2\left(\left.f\right|_{x=0}\right)=\frac{1}{y^2}\left(\left.f\right|_{x=0}\right).
$$
From Lemma~\ref{lemma:main property of A} it follows that $\left.f\right|_{x=0}=A(y)$. 
\end{proof}

In order to prove equation~\eqref{eq:zero equality2}, we have to prove that, for any $n\ge 0$, we have
$$
\left.\Coef_{y^0}\left(D\cdot\der{^n}{x^n}S_y f\right)\right|_{x=0}=0.
$$
By Lemmas \ref{lemma:dx-derivative} and \ref{lemma:specialization}, we have
$$
\left.\Coef_{y^0}\left(D\cdot\der{^n}{x^n}S_y f\right)\right|_{x=0}=\Coef_{y^0}(D\cdot S_y^{n+1}A(y)).
$$
Let $\hS_y:=-\frac{y^3}{u^2}\d_y-\frac{3y^2}{2u^2}+\frac{1}{y}$.
\begin{lemma}\label{lemma:move the operator}
For any two Laurent series $p,q\in\mbC[u,u^{-1}][y^{-1},y]]$ we have
$$
\Coef_{y^0}\left(p\cdot S_y q\right)=\Coef_{y^0}\left(\hS_y p\cdot q\right).
$$
\end{lemma}
\begin{proof}
We have
\begin{multline*}
\Coef_{y^0}\left(p\cdot\frac{y^3}{u^2}\der{q}{y}\right)=\Coef_{y^{-1}}\left(p\cdot\frac{y^2}{u^2}\der{q}{y}\right)=\\
=\Coef_{y^{-1}}\left(\der{}{y}\left(p\cdot\frac{y^2}{u^2}q\right)\right)-\Coef_{y^{-1}}\left(\der{}{y}\left(\frac{y^2}{u^2}p\right)\cdot q\right)=\Coef_{y^0}\left(\left(-\frac{y^3}{u^2}\der{}{y}-\frac{2y^2}{u^2}\right)p\cdot q\right).
\end{multline*}
The lemma is clear now.
\end{proof}
Using this lemma we get
$$
\Coef_{y^0}(D\cdot S_y^{n+1}A(y))=\Coef_{y^0}(\hS_yD\cdot S_y^n A(y))\stackrel{\text{by Lemma~\ref{lemma:equation for D}}}{=}\Coef_{y^0}\left(\frac{A(-y)}{y}\cdot S_y^n A(y)\right).
$$
Define the power series $B(y)$ by $B(y):=yS_y A(y)$. It is easy to see that
$$
B(y)=1+\sum_{i\ge 1}b_i\frac{y^{3i}}{u^{2i}},\quad\text{where}\quad b_i=-\frac{6i+1}{6i-1}a_i.
$$
By Lemma~\ref{lemma:QR coefficients}, we have
\begin{gather}\label{eq:tmp equation}
\Coef_{y^0}\left(\frac{A(-y)}{y}\cdot S_y^n A(y)\right)=\Coef_{y^0}\left(\frac{Q_n}{y}A(-y)A(y)+\frac{R_n}{y^2}A(-y)B(y)\right),
\end{gather}
where $Q_n,R_n\in\mbC[u^{-1},y^{-1}]$ are even Laurent polynomials in $y$. Since $A(-y)A(y)$ is even, we get 
$\Coef_{y^0}\left(\frac{Q_n}{y}A(-y)A(y)\right)=0$. We also have the following identity (see e.g.~\cite[page 36]{PPZ13}): 
$$
A(y)B(-y)+A(-y)B(y)=2.
$$
Therefore, $\Coef_{y^0}\left(\frac{R_n}{y^2}A(-y)B(y)\right)=0$. We conclude that the right-hand side of \eqref{eq:tmp equation} is equal to zero. The proposition is proved.
\end{proof}

Let us prove Theorem~\ref{theorem:main}. By equation~\eqref{eq:equations for the wave function} and Theorem~\ref{theorem:open kdv and the wave function}, the both series $\Coef_{z^0}(D\cdot\psi)$ and $\exp(F^{o,ext})$ are solutions of system~\eqref{eq:equations for F}. Therefore it is sufficient to check that 
$$
\left.\Coef_{z^0}(D\cdot\psi)\right|_{t_{>0}=0}=\left.\exp(F^{o,ext})\right|_{t_{>0}}.
$$
Since the potential $F^{o,ext}$ satisfies the initial condition~\eqref{eq:extended initial condition}, it remains to prove that 
$$
\left.\Coef_{z^0}(D\cdot\psi)\right|_{t_{>0}=0}=1.
$$
We can easily see that $\left.\Coef_{z^0}(D\cdot\psi)\right|_{t_*=0}=1$. Proposition~\ref{proposition:main proposition} and the string equation~$L_{-1}\tau^c=0$ imply that
$$
\left(-\frac{\d}{\d t_0}+\sum_{n\in\half\Z_{\ge 0}}t_{n+1}\der{}{t_n}+u^{-1}t_\half\right)\Coef_{z^0}(D\cdot\psi)=0.
$$
Thus, $\frac{\d}{\d t_0}\left(\left.\Coef_{z^0}(D\cdot\psi)\right|_{t_{>0}=0}\right)=0$, and the theorem is proved.

Obviously, Proposition~\ref{proposition:main proposition} and Theorem~\ref{theorem:main} imply the following corollary.
\begin{corollary}\label{corollary:extended open string}
We have $\cL_{-1}^{ext}\exp(F^{o,ext}+F^c)=0$.
\end{corollary}


\section{Extended open Virasoro equations}\label{section:virasoro equations}

In this section we prove Theorem~\ref{theorem:virasoro} and show that the open Virasoro equations~\eqref{eq:open Virasoro equations} can be easily derived from it.

\subsection{Proof of Theorem~\ref{theorem:virasoro}}

By Corollary~\ref{corollary:extended open string}, we already know that $\cL_{-1}^{ext}\exp(F^{o,ext}+F^c)=0$. By Theorem~\ref{theorem:main}, we have $\exp(F^{o,ext}+F^c)=\Coef_{z^0}(D\cdot H)$. A direct computation shows that, for any $n\ge -1$, we have
\begin{gather*}
\cL_n^{ext}H=G_z(L_n\tau^c)\exp(\xi)+\frac{u^{2n+2}z^{2n+2}}{2^{n+1}}\left(\cL^{ext}_{-1}H-G_z(L_{-1}\tau^c)\exp(\xi)\right)+\frac{n+1}{2^n}u^{2n}z^{2n}H.
\end{gather*}
By equation~\eqref{eq:second equation for the wave function}, $L\psi=z^2\psi$. Using also the closed Virasoro equations~\eqref{eq:closed virasoro}, we get
\begin{gather}
\cL_n^{ext}H=\frac{u^{2n+2}}{2^{n+1}}\left(\cL_{-1}^{ext}\circ\tau^c\circ L^{n+1}\circ(\tau^c)^{-1}\right)H+\frac{u^{2n}(n+1)}{2^n}\left(\tau^c\circ L^n\circ(\tau^c)^{-1}\right)H\label{eq:tmp1}.
\end{gather}
Let $\tcL_{-1}^{ext}:=\cL_{-1}^{ext}-\frac{x^2}{2u^2}$. From the string equation~$L_{-1}\tau^c=0$ it follows that
\begin{gather*}
\cL_{-1}^{ext}\circ\tau^c=\tau^c\circ\tcL_{-1}^{ext},\qquad \left[\tcL_{-1}^{ext},L\right]=-\frac{2}{u^2}.
\end{gather*}
Combining these equations with equation~\eqref{eq:tmp1}, we get
\begin{align*}
\cL_n^{ext}H=&\frac{u^{2n+2}}{2^{n+1}}\left(\tau^c\circ\tcL_{-1}^{ext}\circ L^{n+1}\circ(\tau^c)^{-1}\right)H+\frac{u^{2n}(n+1)}{2^n}\left(\tau^c\circ L^n\circ(\tau^c)^{-1}\right)H=\\
=&\frac{u^{2n+2}}{2^{n+1}}\left(\tau^c\circ L^{n+1}\circ\tcL_{-1}^{ext}\circ(\tau^c)^{-1}\right)H=\frac{u^{2n+2}}{2^{n+1}}\left(\tau^c\circ L^{n+1}\circ(\tau^c)^{-1}\right)\cL_{-1}^{ext}H.
\end{align*}
Since $\exp(F^{o,ext}+F^c)=\Coef_{z^0}(D\cdot H)$, we obtain
$$
\cL_n^{ext}\exp(F^{o,ext}+F^c)=\frac{u^{2n+2}}{2^{n+1}}\left(\tau^c\circ L^{n+1}\circ(\tau^c)^{-1}\right)\cL_{-1}^{ext}\exp(F^{o,ext}+F^c)=0.
$$
The extended open Virasoro equations are proved.


\subsection{Open Virasoro equations from the extended open Virasoro equations}

Let us derive the open Virasoro equations~\eqref{eq:open Virasoro equations} from the extended open Virasoro equations~\eqref{eq:new virasoro}. We have $F^o=\left.F^{o,ext}\right|_{s_{\ge 1}=0}$. From Theorem~\ref{theorem:open kdv and the wave function} it follows that
$$
\frac{\d}{\d s_n}\exp(F^{o,ext})=\frac{u^n}{(n+1)!}\frac{\d^{n+1}}{\d s^{n+1}}\exp(F^{o,ext}).
$$
We see that
$$
\left.\left(\cL^{ext}_n\exp(F^{o,ext}+F^c)\right)\right|_{s_{\ge 1}=0}=\left.\left(\cL_n\exp(F^{o,ext}+F^c)\right)\right|_{s_{\ge 1}=0}=\cL_n\exp(F^o+F^c).
$$
Theorem~\ref{theorem:virasoro} implies that $\cL_n\exp(F^o+F^c)=0$. The open Virasoro equations are proved.


{
\appendix

\section{Technical lemmas}\label{section:technical lemmas}

In this section we collect several technical facts and lemmas that we used in Sections~\ref{section:open kdv and wave function} and~\ref{section:main theorem}. In Section~\ref{subsection:appendix,tau-function} we list some properties of a tau-function of the KdV hierarchy. Section~\ref{subsection:A,B and D} is devoted to properties of the power series $A,B$ and $D$.

\subsection{Tau-function}\label{subsection:appendix,tau-function}

Let $L=\d_x^2+2w$ be a solution of the Lax equations~\eqref{eq:Lax equations} and~$\tau$ be the tau-function. Denote the logarithm $\ln\tau$ by $F$.

\subsubsection{Second derivatives $F_{i,j}$}

Suppose $f$ is a power series in $t_0,t_1,\ldots$ with the coefficients from $\mbC[u,u^{-1}]$. We will say that the power series $f$ is a differential polynomial in $w$, if it can be expressed as a polynomial in $w,w_x,w_{xx},\ldots$ with the coefficients from $\mbC[u,u^{-1}]$. We will sometimes denote $\d_x^i w$ by $w_i$.

\begin{lemma}[See e.g. \cite{Dic03}]
For any $i,j\ge 0$ there exists a universal polynomial $P_{i,j}\in\mbC[u,u^{-1}][v_0,v_1,v_2,\ldots]$ in formal variables $v_0,v_1,v_2,\ldots$ such that 
$$
F_{i,j}=\left.P_{i,j}\right|_{v_k=w_k}.
$$
The polynomial $P_{i,j}$ is universal in the sense, that it doesn't depend on a solution $L$ of the Lax equations and on the choice of a $\tau$-function. The polynomials $P_{i,j}$ have the property
\begin{gather}\label{eq:zero}
\left.P_{i,j}\right|_{v_*=0}=0.
\end{gather}
\end{lemma} 
\noindent For example, $P_{0,0}=v_0$. Therefore, $F_{0,0}=w$.

\subsubsection{Recursion for the operators $(L^{n+\half})_+$}

\begin{lemma}
For any $n\in\Z_{\ge 0}$ we have
\begin{gather*}
L^{n+\frac{1}{2}}=\left(L^{n+\frac{1}{2}}\right)_++\frac{(2n+1)!!}{u^{2n}}F_{0,n}\d_x^{-1}-\frac{1}{2}\frac{(2n+1)!!}{u^{2n}}F_{0,0,n}\d_x^{-2}+\ldots.
\end{gather*}
\end{lemma}
\begin{proof}
Let 
\begin{gather*}
L^{n+\frac{1}{2}}=\left(L^{n+\frac{1}{2}}\right)_++a_n\d_x^{-1}+b_n\d_x^{-2}+\ldots.
\end{gather*}
We have $[L^{n+\frac{1}{2}},L]=0$, therefore,
\begin{gather}\label{eq:eq1}
\left[\left(L^{n+\frac{1}{2}}\right)_++a_n\d_x^{-1}+b_n\d_x^{-2},\d_x^2+2w\right]_+=0.
\end{gather}
Since $w=F_{0,0}$, we have $\frac{u^{2n}}{(2n+1)!!}\left[\left(L^{n+\frac{1}{2}}\right)_+,L\right]=2F_{0,0,n}$. Expanding~\eqref{eq:eq1} we get
$$
2\frac{(2n+1)!!}{u^{2n}}F_{0,0,n}-2\d_x a_n=0.
$$
Note that $a_n$ is a universal differential polynomial in $w,w_x,w_{xx},\ldots$. It is also easy to see that~$a_n$ is equal to zero, if~$w$ is zero. Therefore, we conclude that $a_n=\frac{(2n+1)!!}{u^{2n}}F_{0,n}$.

Let us compute $b_n$. We have
\begin{gather}\label{eq:recursion1}
\left(L^{n+\frac{3}{2}}\right)_+=\left(L^{n+\frac{1}{2}}\right)_+\circ L+a_n\d_x+b_n.
\end{gather}
Note that for any non-negative integer $k$ we have 
$$
\left(L^{\frac{k}{2}}\right)^*_+=(-1)^k\left(L^{\frac{k}{2}}\right)_+.
$$
Thus,
\begin{gather}\label{eq:conjugate}
-\left(L^{n+\frac{3}{2}}\right)_+=-L\circ\left(L^{n+\frac{1}{2}}\right)_+-\d_xa_n-a_n\d_x+b_n.
\end{gather}
Summing~\eqref{eq:recursion1} and~\eqref{eq:conjugate}, we obtain
$$
\left[\left(L^{n+\frac{1}{2}}\right)_+,L\right]-\d_x a_n+2b_n=0.
$$
Thus, $b_n=-\frac{1}{2}\frac{(2n+1)!!}{u^{2n}}F_{0,0,n}$. The lemma is proved.
\end{proof}

\begin{lemma}\label{lemma:recursion lemma}
For any $n\in\Z_{\ge 1}$ we have
$$
\left(L^{n+\half}\right)_+=\left(L^{n-\half}\right)_+\circ L+\frac{(2n-1)!!}{u^{2n-2}}F_{0,n-1}\d_x-\frac{1}{2}\frac{(2n-1)!!}{u^{2n-2}}F_{0,0,n-1}.
$$
\end{lemma}
\begin{proof}
This is just equation~\eqref{eq:recursion1} with $n$ replaced by $n-1$.
\end{proof}


\subsection{Properties of the series $A,B$ and $D$}\label{subsection:A,B and D}

Recall that $S_y:=\frac{y^3}{u^2}\der{}{y}+\frac{y^2}{2u^2}+\frac{1}{y}$, $\hS_y:=-\frac{y^3}{u^2}\der{}{y}-\frac{3y^2}{2u^2}+\frac{1}{y}$, $A(y):=1+\sum_{i\ge 1}a_i\frac{y^{3i}}{u^{2i}}$ and $D:=1+\sum_{i\ge 1}d_i\frac{y^{3i}}{u^{2i}}$.

Let $C$ be a power series of the form $C=1+\sum_{i\ge 1}c_i\frac{y^{3i}}{u^{2i}}$, where $c_i$ are some complex coefficients. Consider the differential equation
\begin{gather}\label{eq:differential equation for A}
S_y^2C=\frac{1}{y^2}C.
\end{gather} 

\begin{lemma}\label{lemma:main property of A}
The series $A(y)$ is a unique solution of equation \eqref{eq:differential equation for A}.
\end{lemma}
\begin{proof}
Equation~\eqref{eq:differential equation for A} implies that 
\begin{gather}\label{eq:tmp3}
\left(\frac{\d^2}{\d y^2}+\left(\frac{4}{y}+\frac{2u^2}{y^4}\right)\frac{\d}{\d y}+\frac{5}{4y^2}\right)C=0.
\end{gather}
The fact, that the series $A(y)$ is a unique solution of~\eqref{eq:tmp3}, can be checked by a simple direct computation.
\end{proof}

\begin{lemma}\label{lemma:equation for D}
The series $D$ satisfies the differential equation $\hS_y D=\frac{A(-y)}{y}$.
\end{lemma}
\begin{proof}
From \eqref{eq:definition of D} it immediately follows that for any $n\ge 1$ we have
$$
d_n=\left(3n-\frac{3}{2}\right)d_{n-1}+|a_n|.
$$
This easily implies the statement of the lemma.
\end{proof}

Recall that the series $B(y)$ is defined by $B(y):=y S_yA(y)$.
\begin{lemma}\label{lemma:square operator}
We have $S_y\left(\frac{B(y)}{y}\right)=\frac{A(y)}{y^2}$.
\end{lemma}
\begin{proof}
The equation is equivalent to the equation $S_y^2 A(y)=\frac{A(y)}{y^2}$, that follows from Lemma~\ref{lemma:main property of A}.
\end{proof}

\begin{lemma}\label{lemma:QR coefficients}
For any $n\ge 0$ we have 
$$
S_y^n A(y)=Q_n A(y)+R_n \frac{B(y)}{y}
$$
for some Laurent polynomials $Q_n,R_n\in\mbC[u^{-1},y^{-1}]$, that are even as functions of $y$.
\end{lemma}
\begin{proof}
Define a sequence of Laurent polynomials $Q_n,R_n\in\mbC[u^{-1},y,y^{-1}]$, $n\ge 0$, by
\begin{align*}
& Q_0=1,                                                        && R_0=0,\\
& Q_{n+1}=\frac{y^3}{u^2}\frac{\d Q_n}{\d y}+\frac{1}{y^2}R_n,  && R_{n+1}=Q_n+\frac{y^3}{u^2}\frac{\d R_n}{\d y}.
\end{align*}
This recursion immediately implies that~$Q_n$ and~$R_n$ are even as functions of~$y$. Using this fact it is easy to see that, actually, $Q_n,R_n\in\mbC[u^{-1},y^{-1}]$. Let us check that
$$
S_y^n A(y)=Q_n A(y)+R_n\frac{B(y)}{y}.
$$
The proof is by induction on~$n$. This equation is obviously true for~$n=0$. For an arbitrary~$n$ we compute
$$
S_y\left(Q_n A(y)+R_n\frac{B(y)}{y}\right)\stackrel{\text{by Lemma~\ref{lemma:square operator}}}{=}\left(\frac{y^3}{u^2}\frac{\d Q_n}{\d y}+\frac{1}{y^2}R_n\right)A(y)+\left(Q_n+\frac{y^3}{u^2}\frac{\d R_n}{\d y}\right)\frac{B(y)}{y}.
$$
The lemma is proved.
\end{proof}

}

\end{document}